\documentclass[proc]{edpsmath}

\usepackage{dsfont}
\usepackage{siunitx}
\usepackage{graphicx}
\usepackage{xcolor}

\newtheorem{theo}{Theorem}[section]
\newtheorem{defi}[theo]{Definition}
\newtheorem{prop}[theo]{Proposition}

\newtheorem{hypo}[theo]{Hypothesis}
\newcommand{\itemr}{\item[$\rightarrow$]}
\newcommand{\itemb}{\item[$\bullet$]}
\newcommand{\R}{\mathbb{R}}
\newcommand{\E}{\mathbb{E}}
\newcommand{\dx}{\textrm{d}}

\begin{document}

\title{A stochastic model for protrusion activity }
\thanks{...}\thanks{...}
\author{Christ\`{e}le Etchegaray }\address{MAP5, CNRS UMR 8145, Universit\'{e} Paris Descartes, 45 rue des Saints P\`{e}res 75006 Paris, France.\\
              \email{christele.etchegaray@parisdescartes.fr}}
\author{Nicolas Meunier}\address{MAP5, CNRS UMR 8145, Universit\'{e} Paris Descartes, 45 rue des Saints P\`{e}res 75006 Paris, France.\\
\email{nicolas.meunier@parisdescartes.fr} 
}

%
%
\begin{abstract} 
In this work we approach cell migration under a large-scale assumption, so that the system reduces to a particle in motion. Unlike classical particle models, the cell displacement results from its internal activity: 
the cell velocity is a function of the (discrete) protrusive forces exerted by filopodia on the substrate. 
Cell polarisation ability is modeled in the feedback that the cell motion exerts on the protrusion rates: faster cells form preferentially protrusions in the direction of motion. 
By using the mathematical framework of structured population processes previously developed to study population dynamics \cite{fournier_microscopic_2004}, we introduce rigorously the mathematical model and we derive some of its fundamental properties. We perform numerical simulations on this model showing that different types of trajectories may be obtained: Brownian-like, persistent, or intermittent when the cell switches between both previous regimes. We find back the trajectories usually described in the literature for cell migration.
\end{abstract}
%
%
%
\maketitle

\selectlanguage{english}

 \section{Introduction}
 
 Cell migration is a fundamental process involved in physiological and pathological phenomena such as the immune response, morphogenesis, but also the development of metastasis from a tumor \cite{Friedl2003Tumour-cell-inv,morphogenesis}. To ensure these functions, cells have a highly complex out-of-equilibrium internal organization where multiscale reactions occur among polymers and molecules, leading to unpredictable macroscopic behaviours. \par 
In the case of cell crawling, cells spread on an adhesive substrate, and form extensions also called protrusions. 
Then, molecular adhesion complexes grow and ensure a mechanical connection between protrusions and the substrate, by which forces are transmitted and lead to a displacement. 
Protrusions of a crawling cell can be divided in two types: \emph{lamellipodia} are wide and flat and fluctuate continuously, while \emph{filopodia} are long finger-like extensions able to grow further and probe the substrate. \par
It has been observed that cells protrusive activity fluctuates a lot, and that these fluctuations are responsible for the long-term characteristics of trajectories \cite{Caballero2014Protrusion-fluc}. 
Cell trajectories can be very different even for a single cell type: some do not explore the environment, while others have a much more efficient displacement. It is of interest to try to capture this diversity in a mathematical model. \par 
Existing stochastic models for cell trajectories are either Random Walks, L\'evy flights or Active Brownian Particle models \cite{Romanczuk2012Active-Brownian}. 
In these models, key-features of the motion are quantified, such as the mean persistence time, or the stationary distribution of the particle's velocity \cite{Romanczuk2012Active-Brownian}. However, the dynamics is macroscopic and processes such as polarisation are taken into account by an arbitrary positive feedback at the scale of the trajectory. In this work, we present a 2D stochastic particle model for cell trajectories based on the filopodial activity, able to reproduce the diversity of trajectories observed.

\section{Model construction}
We choose space and time scales large enough so that the cell is described as an active particle (its center of mass). Cell shape and intracellular dynamics are therefore considered at the cell scale. We define now the velocity model that is based on force equilibrium.

\subsection{Velocity model}
At each time, the cell velocity writes $\vec{V}_{t}$, and its polar coordinates $(v_{t},\theta_{t})$.
The crawling of cells on an adhesive substrate occurs at very small scales. Indeed, cells sizes are of the order of $1-10\,\si{\micro \meter}$, while their speeds ranges at the scale of $\si{\micro \meter \per \second}$. Therefore, inertia is negligible for this system (see more precise justifications of the low Reynolds number setting in e.g \cite{tanimoto2014simple}), and Newton's second law of motion reduces to instantaneous force equilibrium: at all time $t \geq 0$,
\begin{displaymath}
\sum \vec{F}_{ext}(t) =\vec{0}\,.
\end{displaymath}
The cell being an active system, macroscopic forces that apply can be either passive or active. In this case appear
\begin{itemize}
\itemr a passive force: the \textbf{friction force} exerted by the substrate on the cell due to motion, that writes $\vec{f} = -\gamma \vec{V}_t$, with $\gamma$ the global friction coefficient,
\itemr active forces related to the protrusion process. Indeed, a complex internal activity gives rise to forces in the body of the cell. Filopodial protrusions can be considered as good readouts  \cite{Caballero2014Protrusion-fluc}. As a consequence, in the following, only \textbf{filopodial forces} will be considered. Note that at this scale, the formation of filopodia is discontinuous in time.
\end{itemize}

\begin{figure}
\centering
\includegraphics[scale=0.6]{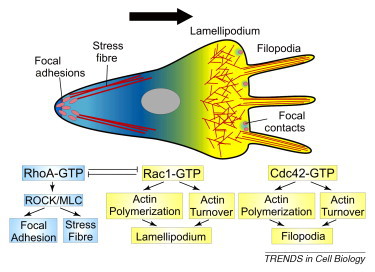}
\caption{Scheme of a polarised crawling cell: protrusive structures at the front (lamellipodium, filopodia), and contractile fibers at the back. The asymmetric activity combines with an asymmetric repartition of molecular regulators organized in feedback loops. Source: \cite{mayor2010keeping} }
\end{figure}
Combining these information, we get: 
\begin{equation}\label{velocity_model}
\gamma \vec{V_t} = \sum_{i=1}^{N_t} \vec{F_i}(t),
\end{equation}
where $N_t$ is the number of filopodia adhering on the substrate at time $t$, and $(\vec{F_i}(t))_{i}$ the filopodial forces. The cell motion is then entirely described by the protrusions.\par

\begin{hypo}\label{hypo2}
Each filopodial force vector is unitary and constant in time.
\end{hypo}
Now, denoting $\theta_i = arg(\vec{F_i})$, one can write
\begin{displaymath}
\vec{F_i}(t) = \begin{pmatrix}
\cos(\theta_i) \\ \sin(\theta_i)
\end{pmatrix}.
\end{displaymath}

Modelling cell motion then accounts to modelling the time evolution of the filopodial population in terms of individual orientation. 

\subsection{Protrusion model}
Each filopodium is characterized by a quantitative parameter, its orientation $\theta \in [0,2\pi)$. Therefore, we use a measure valued process for the stochastic evolution of the set of filopodia, as in ecological population models \cite{fournier_microscopic_2004}.
Let us denote $\mathcal{M}_F(\overline{\chi})$ the set of positive finite measures on $\overline{\chi}=[0,2\pi]$, equipped with the weak topology. Notice that as $\overline{\chi}$ is compact, weak and vague topologies on $\mathcal{M}_F(\overline{\chi})$ coincide. Write 
$\mathcal{M}$ for the subset of $\mathcal{M}_F(\overline{\chi})$ composed of all finite point measures. Then, a filopodium of orientation $\theta$ is described by a Dirac measure $\delta_\theta$ on $\chi$, and the whole population by 
\begin{displaymath}
\nu_t = \sum_{i=1}^{N_t} \delta_{\theta_i} \in \mathcal{M}.
\end{displaymath}

For any measurable function $f$ on $\overline{\chi}$ and any $\mu \in \mathcal{M}_F(\overline{\chi})$, we have $<\mu,f> = \int_{\chi} f(\theta) \mu(\dx \theta)\,.$ In particular, $<\nu_t,f> = \sum_{i=1}^{N_t} f(\theta_i)$, and the population size corresponds to $N_t = <\nu_t,1>$.
A simple way to express the velocity equation (\ref{velocity_model}) together with hypothesis (\ref{hypo2}) is to write

\begin{displaymath}
\gamma \vec{V_t} = \begin{pmatrix}
<\nu_t,\cos> \\ <\nu_t,\sin>
\end{pmatrix}.
\end{displaymath}

The cell motion is entirely described by a measure-valued markovian jump process $(\nu_t)_t$, as in adaptive stuctured population models. We describe now the different events arising. 

\begin{itemize}
\itemb The basic dynamics arising is the \textbf{isotropic appearance} of filopodia. It is responsible for the spontaneous activity that is observed experimentally. We write $\mathbf{c}$ for the creation rate.\par 
\itemb Each filopodium ends up disappearing: the disappearance or \textbf{death} rate is denoted by $\mathbf{d}$ constant.
\itemb \textbf{Polarisation} is characterized by a morphological and functional asymmetry visible both on cell shape and at the microscopic scale \cite{Siam_CHMV,PLOS}. Here, we use the mesoscopic scale of the model to account for polarisation by its feedback on the protrusive activity. Two phenomena have to be distinguished:
\begin{itemize}
\itemr The formation of a protrusion is induced by several microscopic regulators and generates a local positive feedback on the protrusive machinery. Following that, we assume that each filopodium is able to \textbf{reproduce}. Denote $\mathbf{r(\theta_i,\nu_t)}$ the individual reproduction rate of a filopodium of orientation $\theta_i$.
\itemr Polarisation is also reinforced by intracellular actin flows, (see \cite{Maiuri2015Actin_flows_med}). In particular, faster actin flows favor the formation of protrusions in a single stable configuration. Denote $\vec{u}$ for the space-averaged actin flow velocity over the cell. As actin flows are inwardly directed, $-\vec{u}$ characterizes the reinforced direction for protrusions. Moreover, we know that $\vec{V}=-\frac{1}{\alpha} \,\vec{u}$, where $\frac{1}{\alpha}$ depends on the cell type and the experimental setting. Therefore, we consider a positive coupling between the reproduction rate and $-\vec{u} = \alpha \vec{V}$, imposing a global feedback. 
\end{itemize}
\itemb Reproduction of spatially localized filopodia questions the localization of the new protrusions. An individual can reproduce to form a filopodium with the same orientation, or it can have a slightly different location. This phenomenon accounts for the stochastic fluctuations arising in the cell signalling pathways involved in protrusions. In our model, we describe this using the notion of \textbf{heredity} and subsequent \textbf{mutation event} for the orientation of the "offspring". \par 
For simplicity, we assume that at each reproduction event, the mutation probability $\mathbf{\mu}$ is constant. In the case of a mutant new protrusion, its orientation is determined following a probability distribution $\mathbf{g(z;\theta_i)}$ assumed centered in the parent's orientation $\theta_i$, with a constant variance.
\end{itemize}
The possible events are summed up in the following graph:

\begin{center}
\centering
\begin{tabular}{cccclll}
Creation (global) &&&&  \\
\textcolor{blue}{$c$} &&&&\\
   & & Clone &  &  \\
Reproduction (individual) &  $\nearrow$ & \textcolor{blue}{$1-\mu$} & & \\
 \textcolor{blue}{$r(\theta_i, \nu)$} &  $\searrow$ &&&  \\
 && Mutation&  $\longrightarrow$ & Choice of $\theta$ \\
Death (individual)&& \textcolor{blue}{$\mu$} &&\textcolor{blue}{$ g(z ; \theta_i)$}\\
\textcolor{blue}{$d$}&&&&
\end{tabular}
\end{center}

Let us comment on the mathematical features of the model. In the case of no interaction between individuals, or global feedback, the process $(N_t)_t$ simply follows an immigration, birth and death dynamics, and mathematical information can be derived. In particular, the branching property still holds. This is no longer the case when adding interactions. For example, if the interaction relies on $V_t$, then knowing only $N_t$ is not sufficient and one has to know about the structured quantities $(N^{\theta_1},N^{\theta_2},...)$ at all time.\par

\subsubsection*{A choice of reproduction rate and mutation law}\label{ex:taux}
As protrusions are located on $[0,2\pi)$, it is natural to consider reproduction rates as circular functions. We chose a reproduction function that is positively correlated to $\alpha \vec{V}_t$ to account for polarisation. The idea is to have a function centered in $\theta_t$ the direction of motion, that gets sharper with increasing $v_t$.


We choose a reproduction rate as a multiple of a circular normal distribution density written
\begin{displaymath}
f(\theta_i; \vec{V}_t) = \frac{1}{2\pi I_0(\kappa(|| \vec{V}_t ||))} \exp(\kappa(|| \vec{V}_t ||) \cos(\theta_i - \text{arg}(\vec{V}_t))),
\end{displaymath}
with $\kappa (|| \vec{V}_t ||) \geq 0$ a non decreasing shape parameter, and $I_0$ the $0$-order modified Bessel function of the first kind. 
Therefore, we denote $r(\theta,\nu_s) = r^* f(\theta; \theta_t, \kappa)$. 

\begin{figure}%
\centering
\begin{minipage}{0.9\linewidth}%
\centering
\includegraphics[scale=0.5]{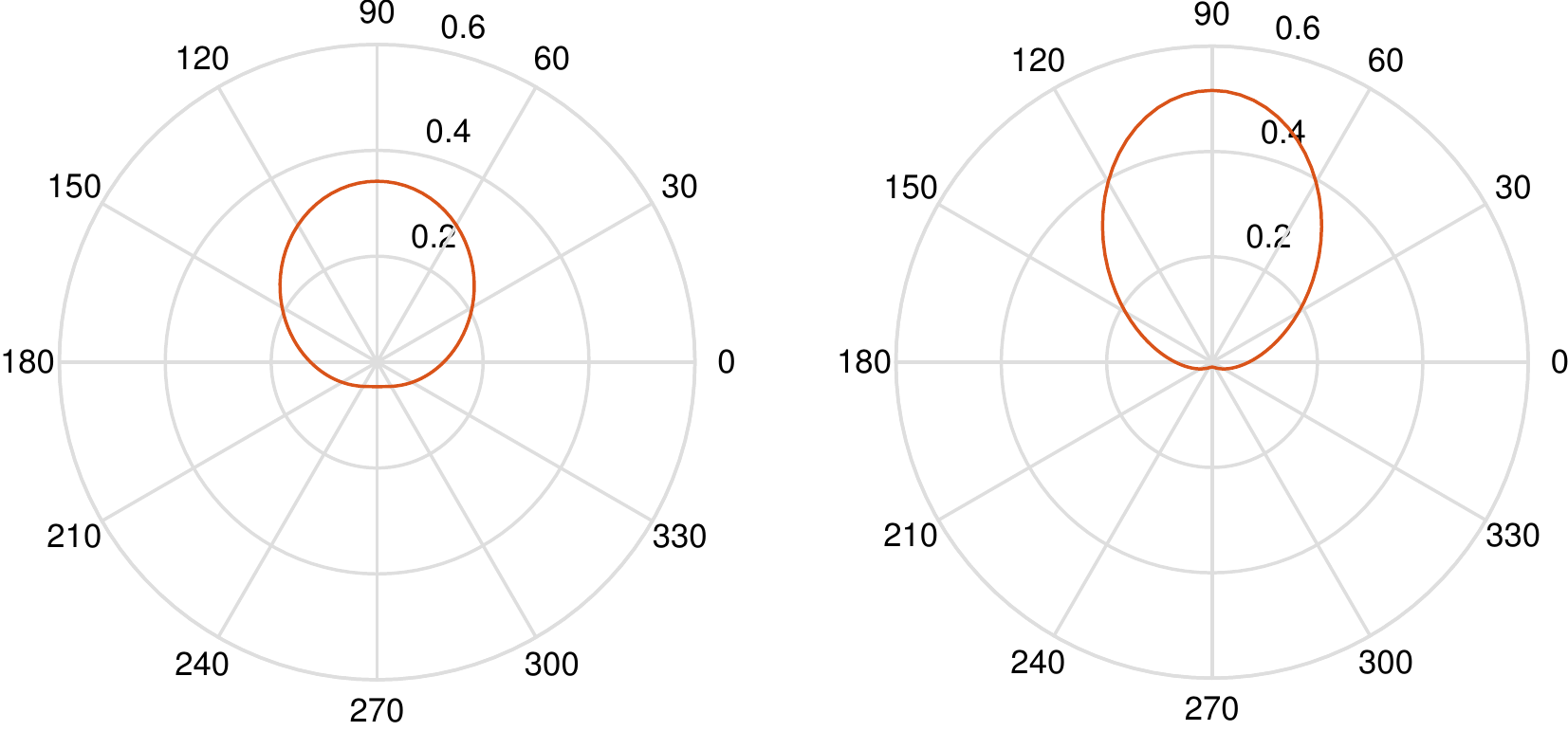}
\caption{Circular distribution for $\theta_t=\frac{\pi}{2}$, $\kappa=1$ (left) and $\kappa=2$ (right).}\label{fig:circular}%
\end{minipage}%
\end{figure}%

The mutation law is a circular normal distribution on $[0,2\pi)$ of density 
 $g(z;\theta_i)$, centered in $\theta_i$ and with a constant shape parameter (resp. variance) $\kappa$ (resp. $\sigma^2$).

\section{Mathematical properties}
From now on, we will use the notation $C$ for any constant, that will change from line to line.

We introduce here a stochastic differential equation for $(\nu_t)_t$ driven by Point Poisson Measures. We will show existence and uniqueness of a solution, and prove that it follows the dynamics previously described. 

In order to pick a specific individual in the population, we have to be able to order them or their trait. Indeed, from the n-uplet $(\theta_1,...,\theta_N)$, one can recover $\nu = \sum_{i=1}^N \delta_{\theta_i}$, but from $\nu$ it is only possible to know $\{\theta_1,...,\theta_N\}$.

\begin{defi}\label{notation}
Let us define the function 
\begin{displaymath}
\begin{array}{lccl}
H=(H^{1},...,H^{k},...) : & \mathcal{M} & \longrightarrow & \left(\chi\right)^{\mathbb{N}^{*}} \\
&&&\\
    &\displaystyle \nu = \sum_{i=1}^{n} \delta_{\theta_{i}} & \longmapsto & (\theta_{\sigma(1)},...,\theta_{\sigma(n)},... ), 
\end{array}
\end{displaymath}
with $\theta_{\sigma(1)} \preceq \theta_{\sigma(2)}\preceq ... \preceq \theta_{\sigma(n)}$, for an arbitrary order $\preceq$. 
Now, an individual can be picked by its label $i$, and the corresponding trait writes $H^{i}(\nu) = \theta_{\sigma(i)}$.
\end{defi}

Let $(\Omega,\mathcal{F},\mathbb{P})$ be a probability space, and $n(\dx i)$ the counting measure on $\mathbb{N}^*$. We introduce the following objects: 

\begin{itemize}\label{objets_proba}
\itemb $\nu_0\in \mathcal{M}$ the finite point measure describing the initial population, eventually equal to the null measure. It can be chosen stochastic as soon as $\mathbb{E}[<\nu_0,1>]<+\infty$.
\itemb $M_0(\dx s,\dx\theta,\dx u)$ a Poisson Point Measure on $[0,+\infty) \times \chi \times \mathbb{R}_+$, of intensity measure $\dx s \,\dx\theta\, \dx u$,
\itemb $M_1(\dx s,\dx i,\dx u)$ and $M_3(\dx s,\dx i,\dx u)$ Poisson Point Measures on $[0,+\infty) \times \mathbb{N}^* \times \mathbb{R}_+$, both of intensity measure $\dx s\, n(\dx i) \,\dx u$,
\itemb $M_2(\dx s,\dx i,\dx\theta,\dx u)$ a Poisson Point Measure on $[0,+\infty) \times \mathbb{N}^* \times \chi \times \mathbb{R}_+$, of intensity measure $\dx s n(\dx i) \dx\theta \dx u$.
\end{itemize}
The Poisson Measures are independent. Finally, $(\mathcal{F}_t)_{t\geq 0}$ denotes the canonical filtration generated by these objects. Let us construct the  $(\mathcal{F}_t)_{t\geq 0}$-adapted process $(\nu_t)_{t\geq 0}$ as the solution of the following SDE: $\forall t\geq 0,$

\begin{equation}\label{processus}
\begin{array}{lllllll}
\nu_t &=& \nu_0 & &&& \\
 &+&\displaystyle  \int_{0}^t \int_{\chi \times \mathbb{R}_+} &\delta_{\theta} && \displaystyle \mathds{1}_{u \leq \frac{c}{2\pi}} & M_0(\dx s,\dx\theta,\dx u) \\
&+& \displaystyle \int_{0}^t \int_{\mathbb{N}^* \times \mathbb{R}_+} & \delta_{H^i(\nu_{s})} &\displaystyle \mathds{1}_{i \leq N_{s}} & \mathds{1}_{u \leq (1- \mu)r(H^i(\nu_{s}),\nu_s) } &  M_1(\dx s,\dx i,\dx u) \\  
&+& \displaystyle \int_{0}^t \int_{\mathbb{N}^* \times \chi \times \mathbb{R}_+} &  \delta_{\theta} & \displaystyle \mathds{1}_{i \leq N_{s}} &\displaystyle \mathds{1}_{u \leq \mu r(H^i(\nu_{s}),\nu_s)  g\left( \theta; H^i(\nu_{s})\right))} & M_2(\dx s,\dx i,\dx\theta,\dx u) \\
&-&\displaystyle \int_{0}^t \int_{\mathbb{N}^* \times \mathbb{R}_+} &\delta_{H^i(\nu_{s})} & \displaystyle \mathds{1}_{i \leq N_{s}} & \displaystyle \mathds{1}_{u \leq d} & M_3(\dx s,\dx i,\dx u).
\end{array}
\end{equation}

In this equation, each term describes a different event. The Poisson Point Measures generate atoms homogeneously in time. However, the dynamics we want to describe follows state-dependent rates. Hence, we use indicator functions to keep only some of the events in order to get the wanted rates. Then, the Dirac measures correspond to the individuals added to or removed from the population. \par 

\begin{hypo}\label{hypo3}
The reproduction rate is a bounded function: 
\[\exists \overline{r}>0 \textrm{ such that }\forall \nu \in \mathcal{M},  \quad \forall (\theta,\nu) \in \chi \times \mathcal{M}, \quad 0\leq r(\theta,\nu) \leq \overline{r}.\]
\end{hypo}

\subsection{Existence and uniqueness}
In this part let us prove existence and uniqueness of a solution for equation (\ref{processus}). Recall that $N_t = <\nu_t,1>$. 
\begin{prop}\label{existence}
Assume the boundedness of the reproduction rate (hypothesis \ref{hypo3}), and that $\mathbb{E}[N_0]<+\infty$. Then, the two following properties hold.
\begin{enumerate}
\item There exists a solution $\nu \in \mathbb{D}(\mathbb{R}_+,\mathcal{M}(\chi))$ of equation (\ref{processus}) such that
\begin{equation}\label{controlNt}
\forall T>0,\; \mathbb{E}\left[\sup_{t\in [0,T]} N_t\right]<\mathbb{E}\left[ N_0\right]e^{\overline{r}T} + \frac{c}{\overline{r}}(e^{\overline{r}T}-1)<+\infty\,, 
\end{equation}
\item There is strong (pathwise) uniqueness of the solution. 
\end{enumerate}
\end{prop}

\begin{proof}[Proof of \ref{existence}]
The proof is similar to prop. 2.2.5 and 2.2.6 in \cite{fournier_microscopic_2004}.

\begin{enumerate}
\item Let $T_0=0$, and $t\in \mathbb{R}_+$. Then, the global jump rate of $\nu_t$ is smaller than $c + (\overline{r}+d)N_t$. Hence one can $\mathbb{P}-a.s$ define the sequence $(T_k)_{k\in\mathbb{N}^*}$ of jumping times, as well as 
$T_\infty := \lim_{k\rightarrow +\infty} T_k$.\par 

Now, by construction, it is $\mathbb{P}-a.s$ possible to build "step-by-step" a solution of equation (\ref{processus}) on $[0,T_\infty[$. Showing existence of a solution $(\nu_t)_{t\in \mathbb{R}_+} \in \mathbb{D}(\mathbb{R}_+,\mathcal{M}(\chi))$ amounts to showing that $\mathbb{P}-a.s$, $T_\infty = +\infty$. That is equivalent to saying that there cannot be an infinite number of jumps in a finite time interval. \par 

\itemb \textbf{\underline{First, we show the control property (\ref{controlNt})}}. For $n>0$ define the sequence of stopping times $(\tau_n)_n$ by
\begin{displaymath}
\tau_n = \inf_{t \geq 0} \{N_t \geq n\}.
\end{displaymath}

\begin{itemize}
\item[$\rightarrow$] \boldmath \textbf{Let us show that $(\tau_n)_{n\geq 0} $  is a sequence of stopping times for $(\mathcal{F}_t)_t$}. \unboldmath Denote $\sigma_t=\sigma(\nu_s, \; 0\leq s \leq t)$ the $\sigma$-algebra generated by $\{\nu_s,\;0\leq s \leq t\}$. Then $\forall t\geq0$, $\sigma_t \subseteq \mathcal{F}_t$. For $(n,m)\in \left(\mathbb{N}^*\right)^2$, notice that
\begin{displaymath}
\begin{aligned}
\{\tau_n \leq m \} &= \{\inf\{t\geq 0 ,\; \left< \nu_t,1\right> \geq n \} \leq m \} \\
& \in \sigma_m \subseteq \mathcal{F}_m,
\end{aligned}
\end{displaymath}
and $(\tau_n)_{n\geq 0} $ is indeed a sequence of stopping times.
\item[$\rightarrow$]\boldmath \textbf{Now, we prove that for all $T<+\infty$, the quantity 
$\mathbb{E}\left[ \sup_{t\in [0,T \wedge \tau_n]} N_t \right]$ is bounded $\forall n\geq 0$.} \unboldmath \par

For $t\in\mathbb{R}_+$, using equation (\ref{processus}) and dropping the non-positive term, one has
\begin{eqnarray*}
N_{t\wedge \tau_n} &=& <\nu_{t\wedge \tau_n},1> \leq  N_0 + \int_{0}^{t\wedge \tau_n} \int_{\chi \times \mathbb{R}_+} \mathds{1}_{u \leq c}  M_0(\dx s,\dx \theta,\dx u) \\
&+& \int_{0}^{t\wedge \tau_n} \int_{\mathbb{N}^* \times \mathbb{R}_+}  \mathds{1}_{i \leq N_{s}} \mathds{1}_{u \leq (1- \mu)r(H^i(\nu_{s}),\nu_s)} M_1(\dx s,\dx i,\dx u) \\  
&+& \int_{0}^{t\wedge \tau_n}\int_{\mathbb{N}^* \times \chi \times \mathbb{R}_+}  \mathds{1}_{i \leq N_{s}}  \mathds{1}_{u \leq \mu r(H^i(\nu_{s}),\nu_s)  g( \theta; H^i(\nu_{s}))} M_2(\dx s,\dx i,\dx \theta,\dx u).
\end{eqnarray*}
As each integrand is positive, bounded, and integrable with respect to the intensity measure, taking the expectation and using the Fubini theorem, we can write

\begin{eqnarray}\label{inequality1}
\mathbb{E}\left[\sup_{t\in [0,T\wedge \tau_N]}N_t\right] &\leq & \mathbb{E}[N_0] + \E\left[\int_{0}^{T\wedge \tau_N} \left( c + \sum_{i=1}^{N_t} r(\theta_i,\nu_t) \right)\dx t \right]  \nonumber\\
&\leq & \mathbb{E}[N_0] + cT + \overline{r} \int_{0}^{T} \E\left[\sup_{s\in [0,t\wedge \tau_N]}N_s\right] \dx t \nonumber
\end{eqnarray}
leading to the $T$-dependent bound using the Gronwall inequality. 

\item[$\rightarrow$] \boldmath \textbf{ Let us prove that $\mathbb{P}-a.s,\; \lim_{n\rightarrow +\infty} \tau_n = +\infty$.} \unboldmath If this wasn't the case, there would exist $M<+\infty$ and a set $A_M \subset \Omega$ such that $\mathbb{P}(A_M) >0$, and $\forall \omega \in A_M, \; \lim_{n\rightarrow +\infty} \tau_n(\omega) < M$. By the Markov inequality, $\forall T > M,$

\begin{displaymath}
\E\left[\sup_{t\in [0,T\wedge \tau_n]}N_t\right] \geq n \underbrace{\mathbb{P}\left(\sup_{t\in [0,T\wedge \tau_n]}N_t \geq n\right)}_{\geq \mathbb{P}(A_M) >0},
\end{displaymath}
which is in contradiction with equation \eqref{inequality1}. 

\item[$\rightarrow$] \textbf{Property (\ref{controlNt}) is proved by the Fatou lemma:} 
\begin{equation*}
\begin{aligned}
\E\left[ \sup_{t\in [0,T]}N_t \right] &= \E\left[\liminf_{n\rightarrow +\infty} \sup_{t\in [0,T\wedge \tau_n]}N_t \right] \\
&\leq  \liminf_{n\rightarrow +\infty} \E\left[ \sup_{t\in [0,T\wedge \tau_n]}N_t \right] \leq \mathbb{E}[N_0] e^{\overline{r}T} + \frac{c}{\overline{r}}\left( e^{\overline{r}T} -1\right) <+\infty.
\end{aligned}
\end{equation*}

\itemb \textbf{\underline{\boldmath Now, let us show that $\mathbb{P}-a.s$, $T_\infty = +\infty$. \unboldmath}} If this is not the case, then there exists $\overline{M} <+\infty$ and a set $A_{\overline{M}} \subset \Omega$ such that $\mathbb{P}(A_{\overline{M}}) >0$ and $\forall w \in A_{\overline{M}}$, $T_\infty(\omega) < \overline{M}$. Moreover, if the assertion

\begin{equation}\label{assertion1}
\forall \omega \in A_{\overline{M}},\; \lim_{k\rightarrow +\infty} N_{T_k}(\omega) = +\infty,
\end{equation}

is true, then we would have 

\begin{displaymath}
\forall N>0, \; \forall \omega \in A_{\overline{M}},\; \tau_N(\omega) \leq \overline{M},
\end{displaymath}

which contradicts $\lim_{n\rightarrow +\infty} \tau_n = +\infty$. As a consequence, if we prove (\ref{assertion1}), the proposition is proved. 
If (\ref{assertion1}) is not true, there would exist $N'>0$ and a set $B\subset A_{\overline{M}}$ such that $\mathbb{P}(B)>0$ and 

\begin{displaymath}
\forall \omega \in B, \; \forall k \in \mathbb{N},\; N_{T_k}(\omega) < N'.
\end{displaymath}

Then, $\forall \omega \in B$, $(T_k(\omega))_k$ can be seen as the subsequence of a sequence of jumping times $(T_k^1(\omega))_k$ of a Point Poisson Process of intensity $c + (\overline{r}+d)N'$. The only accumulation point of $(T_k^1(\omega))_k$ being $\mathbb{P}-a.s \; +\infty$, it contradicts the definition of $B$, and proves (\ref{assertion1}). 
\end{itemize}

\item The sequence of jumping times $(T_k)_{k\in \mathbb{N}}$ being already defined, we only have to show that $(T_k,\nu_{T_k})_{k\in \mathbb{N}}$ are uniquely determined by $D=(\nu_0,M_0,M_1,M_2,M_3)$ defined above.
But this is clear by construction of the process. 
\end{enumerate}
\end{proof}

\subsection{Markov property}

Now, we can show that the solution $(\nu_t)_t$ of equation (\ref{processus}) is a Markov process  in the Skorohod space $\mathbb{D}(\R_+,\mathcal{M}_F(\chi))$ of càdlàg finite measure-valued processes on $\chi$. For that purpose, we introduce $\forall \nu\in \mathcal{M}$, $\Phi: \mathcal{M} \rightarrow \mathbb{R}$ measurable and bounded, the operator $L$ defined by 

\begin{eqnarray}\label{generator}
L\Phi(\nu) &=& \int_{\chi} \frac{c}{2\pi} \left[ \Phi(\nu + \delta_{\theta}) - \Phi(\nu) \right] \dx \theta \nonumber\\
&+& \int_{\chi} (1- \mu) r(\theta, \nu)  \left[ \Phi(\nu + \delta_{\theta}) - \Phi(\nu) \right] \nu(\dx \theta) \\
&+& \int_{\chi} \mu r(\theta,\nu)  \int_{\chi} \left[ \Phi(\nu + \delta_{z}) - \Phi(\nu) \right] g(z; \theta) \dx z \; \nu(\dx \theta)\nonumber\\
&+& \int_{\chi} d \left[\Phi(\nu - \delta_{\theta}) - \Phi(\nu) \right] \nu(\dx \theta). \nonumber 
\end{eqnarray}

\begin{prop}\label{markovprop}
 Take $(\nu_t)_{t\geq 0}$ the solution of equation (\ref{processus}) with $\mathbb{E}[<\nu_0,1>]<+\infty$. Then, $(\nu_t)_{t\geq 0}$ is a Markovian process of infinitesimal generator $L$.
\end{prop}
In particular, this proposition ensures that the law of $(\nu_t)_{t\geq 0}$ is independent of the order $\preceq$ involved in (\ref{notation}).

\begin{proof}[Proof of proposition \ref{markovprop}]
The process $(\nu_t)_{t\geq 0} \in \mathbb{D}(\mathbb{R}_+,\mathcal{M}(\overline{\chi}))$ is markovian by construction. Now, let $N_0 <N< +\infty$, and consider again the stopping time $\tau_N$. Let $\Phi :  \mathcal{M} \rightarrow \mathbb{R}$ be measurable and bounded. As $\mathbb{P}-a.s$ we can write

\begin{equation}\label{Phi(nu)}
\Phi(\nu_t) = \Phi(\nu_0)  + \sum_{s\leq t} \Phi(\nu_{s^-} + (\nu_{s}-\nu_{s^-}))-\Phi(\nu_{s^-})\,, 
\end{equation}
we have

\begin{eqnarray*}
&& \Phi(\nu_{t\wedge \tau_N}) = \Phi(\nu_{0}) + \int_0^{t\wedge \tau_N} \int_{\chi\times \mathbb{R}_+} \left[ \Phi(\nu_{s^-} + \delta_{\theta}) -\Phi(\nu_{s^-} )  \right] \mathds{1}_{u\leq \frac{c}{2\pi}} M_0(\dx s,\dx \theta,\dx u) \\
&+& \int_0^{t\wedge \tau_N} \int_{\mathbb{N^*}\times \mathbb{R}_+} \left[ \Phi(\nu_{s^-} + \delta_{H^i(\nu_{s^-})}) -\Phi(\nu_{s^-})   \right] \mathds{1}_{i\leq N_{s^-}} \mathds{1}_{u\leq (1-\mu) r(H^i(\nu_{s^-}),\nu_{s^-})} M_1(\dx s,\dx i,\dx u)  \\
&+& \int_0^{t\wedge \tau_N} \int_{\mathbb{N^*}\times  \chi \times \mathbb{R}_+} \left[ \Phi(\nu_{s^-} + \delta_{z}) -\Phi(\nu_{s^-})   \right] \mathds{1}_{i\leq N_{s^-}} \mathds{1}_{u\leq \mu r(H^i(\nu_{s^-}),\nu_{s^-}) g( z; H^i(\nu_{s}))} M_2(\dx s,\dx i,\dx z, \dx u)  \\
&+& \int_0^{t\wedge \tau_N} \int_{\mathbb{N^*} \times \mathbb{R}_+} \left[ \Phi(\nu_{s^-} - \delta_{H^i(\nu_{s^-})}) -\Phi(\nu_{s^-})   \right] \mathds{1}_{i\leq N_{s^-}} \mathds{1}_{u\leq d} M_3(\dx s,\dx i, \dx u)\,.  \\
\end{eqnarray*}

Again, as all integrands are bounded, we can take expectations to get

\begin{eqnarray*}
\mathbb{E}\left[ \Phi(\nu_{t\wedge \tau_N}) \right] &=& \mathbb{E}\left[ \Phi(\nu_{0})  \right] + \mathbb{E}\left[  \int_0^{t\wedge \tau_N} \int_{\chi} \left[ \Phi(\nu_{s^-} + \delta_{\theta}) -\Phi(\nu_{s^-} )  \right] \frac{c}{2\pi} \dx \theta \dx s \right] \\
&+& \mathbb{E}\left[  \int_0^{t\wedge \tau_N} \sum_{i=1}^{N_{s^-}}
 \left[ \Phi(\nu_{s^-} + \delta_{H^i(\nu_{s^-})}) -\Phi(\nu_{s^-})   \right]  (1-\mu) r(H^i(\nu_{s^-}),\nu_{s^-}) \dx s  \right]\\
&+& \mathbb{E}\left[  \int_0^{t\wedge \tau_N} \sum_{i=1}^{N_{s^-}}
\mu r(H^i(\nu_{s^-}),\nu_{s^-}) \int_{\chi } \left[ \Phi(\nu_{s^-} + \delta_{z}) -\Phi(\nu_{s^-})   \right]  g(z; H^i(\nu_{s}))  \dx z \dx s  \right]\\
&+& \mathbb{E}\left[  \int_0^{t\wedge \tau_N}  \sum_{i=1}^{N_{s^-}} \left[ \Phi(\nu_{s^-} - \delta_{H^i(\nu_{s^-})}) -\Phi(\nu_{s^-})   \right] 
d \dx s \right], \\
&=:& \mathbb{E}\left[ \Phi(\nu_{0}) \right] + \mathbb{E}\left[ \psi(t\wedge \tau_N,\nu) \right].
\end{eqnarray*}
On the one hand, $\forall t\in [0,T]$, 
\begin{eqnarray*}
\parallel \psi(t\wedge \tau_N,\nu) \parallel_{\infty} &\leq & 2 T \parallel \Phi \parallel_{\infty} c  + 2 T \parallel \Phi \parallel_{\infty} (1- \mu) \overline{r} N + 2 T \parallel \Phi \parallel_{\infty} \mu \overline{r} N + 2 T \parallel \Phi \parallel_{\infty} d N \\
&\leq & C T \parallel \Phi \parallel_{\infty} (c + (\overline{r}+d)N) < +\infty.
\end{eqnarray*}

On the other hand, $t\mapsto \psi(t\wedge \tau_N,\nu)$ is derivable in $t=0$ $\mathbb{P}-a.s$ (as $\nu \in \mathbb{D}(\mathbb{R}_+,\mathcal{M}(\chi))$), and for a given $\nu_0$, we have

\begin{eqnarray*}
\frac{\partial \psi}{\partial t}(0,\nu_0) &=&  \int_{\chi} \left[ \Phi(\nu_{0} + \delta_{\theta}) -\Phi(\nu_{0} )  \right] \frac{c}{2\pi} \dx \theta \\
&+& \sum_{i=1}^{N_{0}}
 \left[ \Phi(\nu_{0} + \delta_{H^i(\nu_{0})}) -\Phi(\nu_{0})   \right]  (1-\mu) r(H^i(\nu_{0}),\nu_{0})   \\
&+&  \sum_{i=1}^{N_{0}} \mu r(H^i(\nu_{0}),\nu_{0})
\int_{\chi } \left[ \Phi(\nu_{0} + \delta_{z}) -\Phi(\nu_{0})   \right]  g( z; H^i(\nu_{0}))  \dx z  \\
&+&  \sum_{i=1}^{N_{0}} \left[ \Phi(\nu_{0} - \delta_{H^i(\nu_{0})}) -\Phi(\nu_{0})   \right] d\,.
\end{eqnarray*}

Moreover, $\parallel \frac{\partial \psi}{\partial t}(0,\nu_0) \parallel \leq C \parallel \Phi \parallel_{\infty}(c+ N_0 (\overline{r}+d))$. Now,

\begin{displaymath}
\begin{aligned}
L\phi(\nu_0) &:= \left. \frac{\partial \mathbb{E}\left[\phi(\nu_t)\right]}{\partial t}\right\vert_{t=0} \\
&= \int_{\chi} \left[ \Phi(\nu_{0} + \delta_{\theta}) -\Phi(\nu_{0} )  \right] \frac{c}{2\pi} \dx \theta + \sum_{i=1}^{N_{0}}
 \left[ \Phi(\nu_{0} + \delta_{H^i(\nu_{0})}) -\Phi(\nu_{0})   \right]  (1-\mu) r(H^i(\nu_{0}),\nu_{0})  \\
 & +  \sum_{i=1}^{N_{0}} \mu r(H^i(\nu_{0}),\nu_{0})
\int_{\chi } \left[ \Phi(\nu_{0} + \delta_{z}) -\Phi(\nu_{0})   \right] g( z; H^i(\nu_{0}))  \dx z
\\
& +  \sum_{i=1}^{N_{0}} \left[ \Phi(\nu_{0} - \delta_{H^i(\nu_{0})}) -\Phi(\nu_{0})   \right] d,
\end{aligned}
\end{displaymath}
or equivalently 

\begin{displaymath}
\begin{aligned}
L\phi(\nu_0) &= \int_{\chi} \left[ \Phi(\nu_{0} + \delta_{\theta}) -\Phi(\nu_{0})  \right] \frac{c}{2\pi} \dx \theta + \sum_{i=1}^{N_{0}}
 \left[ \Phi(\nu_{0} + \delta_{H^i(\nu_{0})}) -\Phi(\nu_{0})   \right]  r(H^i(\nu_{0}),\nu_{0})  \\
 & +  \sum_{i=1}^{N_{0}} \left[ \Phi(\nu_{0} - \delta_{H^i(\nu_{0})}) -\Phi(\nu_{0})  \right] d\\
 &+ \mu \sum_{i=1}^{N_{0}} r(H^i(\nu_{0}),\nu_{0}) \left( \int_{\chi } \Phi(\nu_{0} + \delta_{z}) g( z; H^i(\nu_{0}))  \dx z - \Phi(\nu_{0} + \delta_{H^i(\nu_{0})}) \right) \,.
 \end{aligned}
\end{displaymath}
\end{proof}

\section{Numerical simulations}

The construction of the process $(\nu_t)_t$ furnishes directly an algorithm for simulations. We proceed as follows: start with the population measure $\nu_k$ at time $t_k$, for a particle located at $X_{k}$.
\begin{description}
\item[Time of next event] let $\tau= c+ <\nu_k,r+d>$ denote the global jump rate of the process. Then, the time of the next event writes $t_{k+1} := t_k + \Delta t $, where 
\begin{displaymath}
\Delta t \sim Exp(\tau)\,.
\end{displaymath}
\item[Nature of the event] what happens at time $t_{k+1}$ is determined as follows:
\begin{itemize}
\item creation of a protrusion occurs with probability $\frac{c}{\tau}$. Its orientation is chosen uniformly on $[0,2\pi)$. 
\item reproduction of the protrusion number $i$ occurs with probability $\frac{r(H^i(\nu_k),\nu_k)}{\tau}$. Then,
\begin{itemize}
\item[$\rightarrow$] with probability $(1-\mu)$, the new protrusion has orientation $H^i(\nu_k)$,
\item[$\rightarrow$] with probability $\mu$, its orientation is chosen with the realization of a random variable having a probability density $g(\cdot; H^i(\nu_k),\nu_k)$.
\end{itemize}
\item protrusion number $i$ disappears with probability $\frac{d}{\tau}$.
\end{itemize} 
The measure $\nu_{k+1}$ is then obtained from $\nu_k$ and the information of the event occuring at time $t_{k+1}$.
\item[Updates] the particle's new position is 
\begin{displaymath}
X_{k+1} = X_k + \Delta t \, V_k\,,
\end{displaymath}
while $V_{k+1}= \frac{1}{\gamma} \begin{pmatrix}
<\nu_{k+1},\cos>\\<\nu_{k+1},\sin>
\end{pmatrix}$.
\end{description}

One only has to start again to get a trajectory over time. 

\subsection{Results}
Let us now present the numerical trajectories we obtained, that are displayed in figures \ref{fig:peu_mut} and \ref{fig:mut}. Recall that polarisation is quantified by $-\alpha v$, so that the larger $\alpha$ is, the more concentrated in the direction of motion the protrusions are formed. We observe indeed different types of trajectories for varying $\alpha$, from Brownian-like to persistent. In figure \ref{fig:peu_mut}, the mutation probability is $\mu=0.2$, whereas it is $\mu=0.8$ in figure \ref{fig:mut}. We observe that the territory exploration is significantly lower for a higher mutation probability. This shows that the mutation events can not be neglected. 

\begin{figure}[p]
\centering
\includegraphics[scale=0.5]{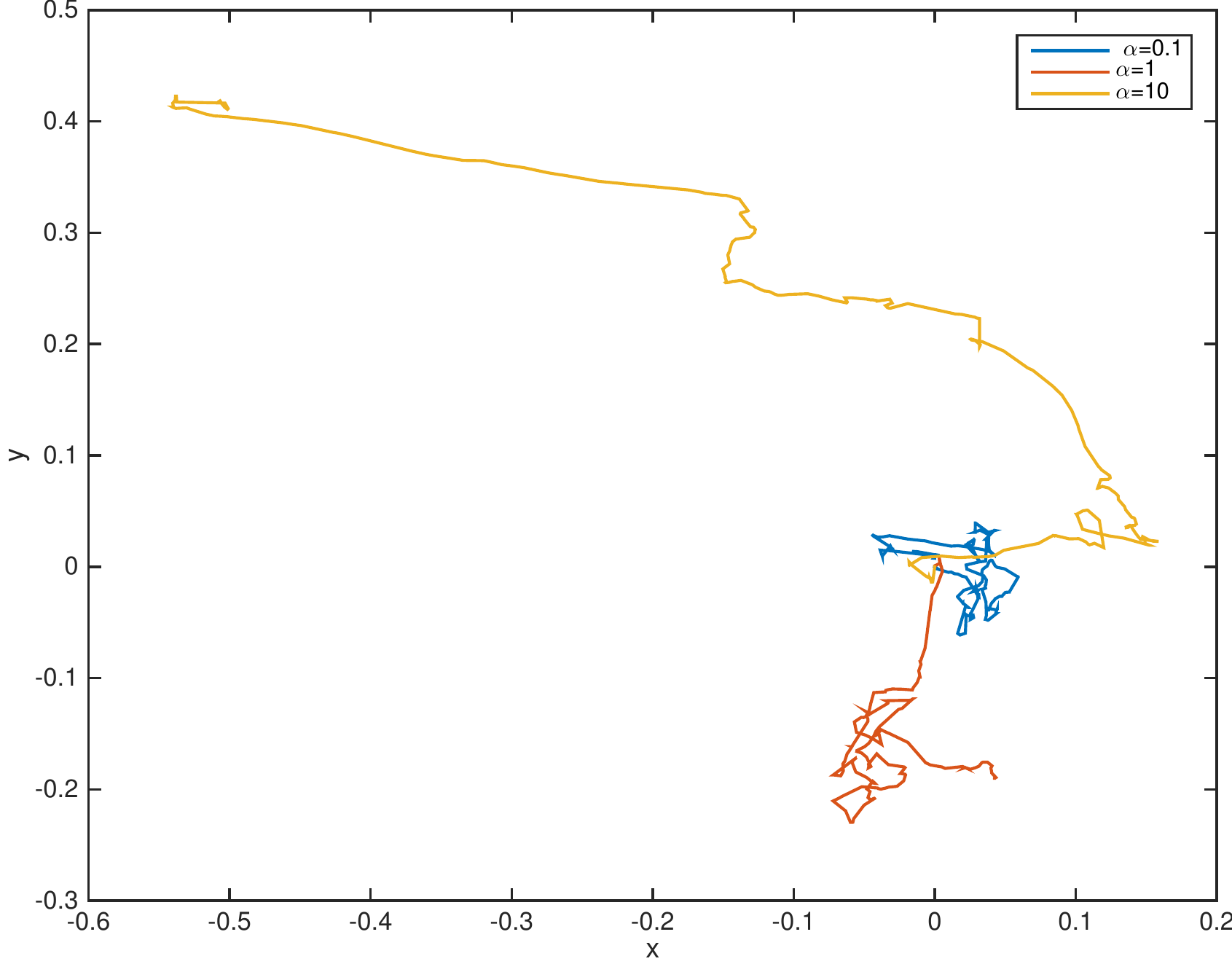}
\caption{Numerical trajectories obtained for a varying polarisation parameter $\alpha$. Parameters: $T=100$, $\Delta t=10^{-4}$, $c=d=1$, $r=0.95$, $\gamma=90$, $\mu=0.2$. Mutation concentration parameter $k=10$.}\label{fig:peu_mut}
\end{figure}

\begin{figure}[p]
\centering
\includegraphics[scale=0.5]{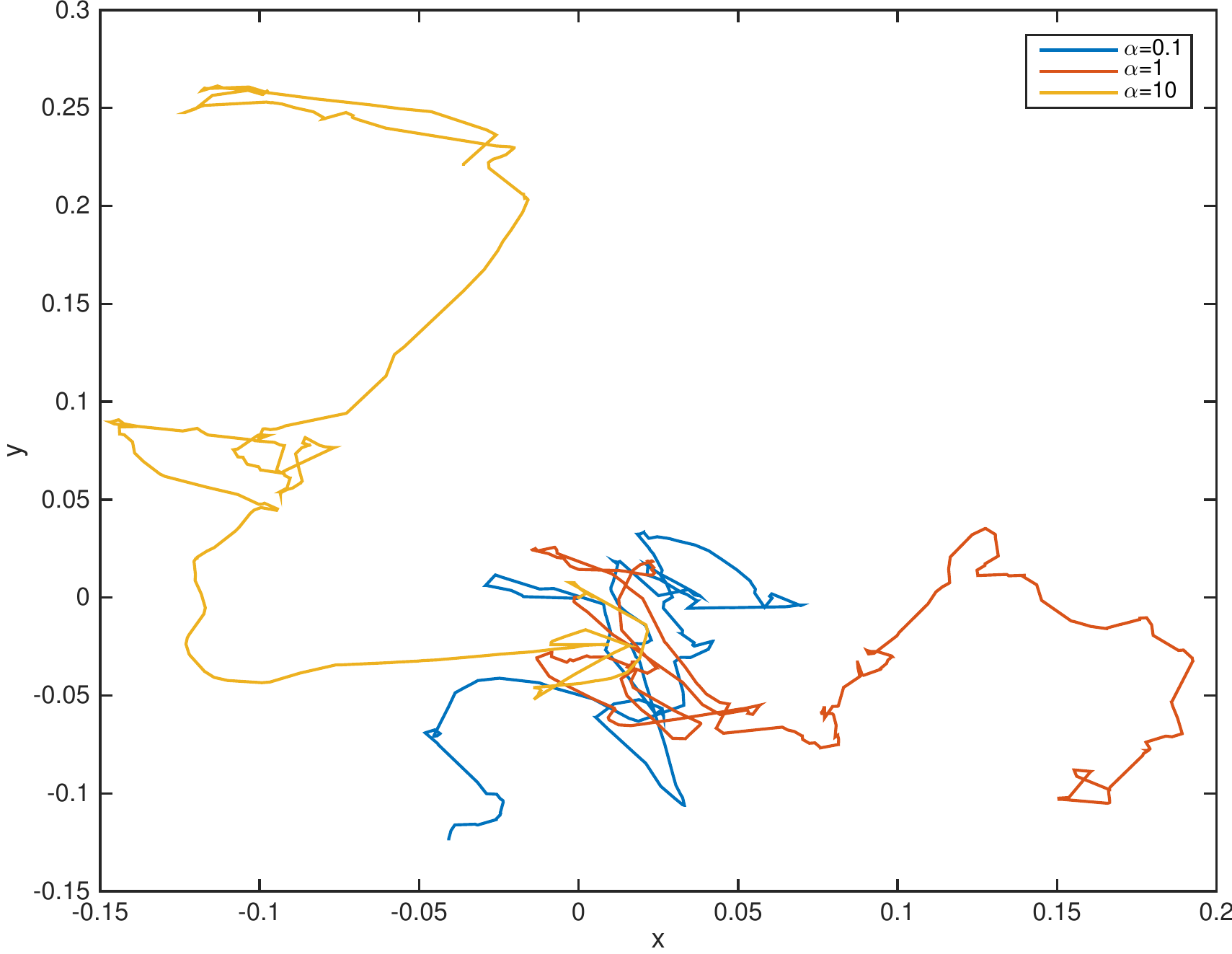}
\caption{Numerical trajectories obtained for a varying polarisation parameter $\alpha$. Parameters: $T=100$, $\Delta t=10^{-4}$, $c=d=1$, $r=0.95$, $\gamma=90$, $\mu=0.8$.}\label{fig:mut}
\end{figure}

\newpage
\bibliographystyle{apalike}
\bibliography{Proc_traj}
\end{document}